\documentclass[10pt,final,conference,letterpaper,twocolumn]{IEEEtran}
\usepackage[dvips]{graphicx}
\usepackage{cite}
\usepackage{epsfig}
\usepackage{amsmath,amssymb,amsfonts,amstext,amsbsy,amsopn,eucal,dsfont}
\usepackage{sublabel}

\newtheorem{proposition}{\textbf{Proposition}}
\newtheorem{lemma}{\textbf{Lemma}}

\newtheorem{definition}{\textbf{Definition}}

\newcommand{\etal}{\emph{et al.}}
\newcommand{\defn}{\triangleq}
\newcommand{\ie}{i.e., }
\newcommand{\dotleq}{\overset{\cdot}{\leq}}

\begin{document}

\title{Diversity-Multiplexing Tradeoff of Network Coding with Bidirectional Random Relaying}

\author{Chun-Hung Liu and Jeffrey G. Andrews \\Department of Electrical and Computer Engineering \\The University of Texas at Austin\\ Austin, TX 78712-0204, USA\\ Email: chliu@mail.utexas.edu and jandrews@ece.utexas.edu}

\maketitle

\begin{abstract}
This paper develops a diversity-multiplexing tradeoff (DMT) over a bidirectional random relay set in a wireless network where the distribution of all nodes is a stationary Poisson point process. This is a nontrivial extension of the DMT because it requires consideration of the cooperation (or lack thereof) of relay nodes, the traffic pattern and the time allocation between the forward and reverse traffic directions.  We then use this tradeoff to compare the DMTs of traditional time-division multihop (TDMH) and network coding (NC). Our main results are the derivations of the DMT for both TDMH and NC.  This shows, surprisingly, that if relay nodes collaborate NC does not always have a better DMT than TDMH since it is difficult to simultaneously achieve bidirectional transmit diversity for both source nodes.  In fact, for certain traffic patterns NC can have a worse DMT due to suboptimal time allocation between the forward and reverse transmission directions.
\end{abstract}


\section{Introduction}
The fundamental tradeoff between diversity and multiplexing gain for point-to-point multiple input and multiple output (MIMO) channels was found in \cite{LZDNCT03}, and has become a popular metric for comparing transmission protocols.  In this work, our first objective is to extend the DMT to the scenario of a multihop bidirectional relaying wireless network.  Our second objective is to apply this to the specific comparison of traditional time-division relaying and network coding, with the goal of learning when or how to use each of those protocols to acquire a better DMT.

Our model considers the practical situation of two communicating nodes in an ad hoc network, whereby each is both the source and the destination for the other. These nodes pairs wish to exchange their packets over one or more relay nodes because the direct channel between them is weak.  There are many approaches to exchanging information between the two nodes, but in order to investigate a non-trivial DMT problem, we assume that the same frequency band is used in both directions and that all nodes are half-duplex, \ie cannot transmit and receive simultaneously.  Specifically, we consider two multihop transmission protocols. The first approach is the traditional approach whereby the two sources share the relays in time.  This so-called time-division multihop (TDMH) approach requires four time slots to exchange a packet in each direction. The second approach is multihop network coding (NC)\footnote{In this paper, we only discuss the DMT problem of network coding with XORing on the MAC layer, which is so called ``digital network coding''. The DMT problem of analog network coding is out of the scope in this paper.} \cite{SYRLRWYNC03,SKHRWHDKMMJC08,CHLFXJGA09}, which is known to be more efficient than TDMH, and indeed saves one time slot compared to TDMH \cite{ABAKDPRAL06}.  Both of the approaches are illustrated in Fig. \ref{Fig:TwoWayMultiRelaySys}(b).

The idea of wireless NC descends from Ahlswede \etal \cite{RANCSYRLRWY00} for improving the capacity of wired networks. By taking advantage of the broadcast nature of the wireless medium, NC achieves a significant throughput gain under certain circumstances\cite{PLNJKES06,PPHY0607,CHLFXJGA09}. It also can be used to exploit cooperative diversity between source and destination nodes\cite{LXTEFJKDJC07,YCSKJL06}. Since NC is able to provide diversity gain as well as throughput gain, it motivates study on how the DMT of NC behaves and if it has better tradeoff compared to TDMH. For example, does the above noted throughput gain of NC come at the expense of diversity
gain? Importantly, we consider bidirectional transmission over a random number of relays -- the nodes in the networks form a stationary Poisson point process and there exits a random set of idle nodes which can assist to route packets between the two source nodes. This plurality of relays may cooperate in a number of different ways or not at all, and each cooperation scenario leads to a different DMT result for both NC and TDMH.

The key to deriving the DMT of TDMH and NC is a suitably defined outage event, defined as a failure of information exchange between the two source nodes. The DMTs of TDMH and NC here are quite different from the previous multihop DMT works (typically see \cite{DGAGHVP08,RVRWH08,ABAKDPRAL06,JNLGWW03} and the references therein) due to their dependence on the traffic pattern, time allocation of bidirectional transmission, as well as the average number of available relay nodes in the random relay set. The main results of this paper are two propositions which respectively provide the DMTs of TDMH and NC. These propositions demonstrate that NC does not always provide a better DMT than TDMH in the relay collaboration case because bidirectional transmit diversity cannot be exploited simultaneously: using an optimally selected relay node to receive and transmit (or broadcast) is practically preferable since it achieves the same DMT and no relay coordination is required. NC could in fact have a worse DMT if there is suboptimal time allocation for a certain traffic pattern. Intuitively, if the offered traffic load is much higher in the forward direction than the reverse direction relative to one of the source nodes, then bidirectional network coding may not be helpful for that source since it presumes a symmetric data rate.

\begin{figure*}[!t]
\centering
\includegraphics[width=6in, height=1.5in]{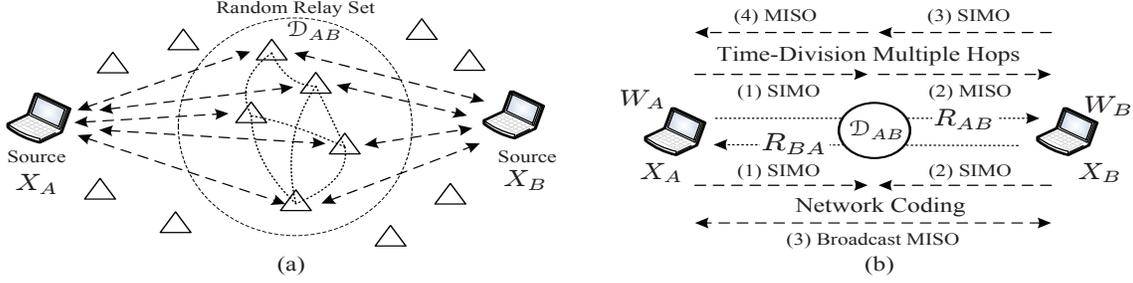}\\
\caption{(a) The bidirectional relaying system: the relay nodes in the random relay set $\mathcal{D}_{AB}$ between source nodes $X_A$ and $X_B$ are a stationary PPP of intensity $\lambda_r$. (b) The equivalent
model. Information exchange between source node A and B is through the intermediate relay
node set $\mathcal{D}_{AB}$. $R_{AB}$ and $R_{BA}$ denote the end-to-end \emph{forward}
and \emph{backward} rates, respectively.}\label{Fig:TwoWayMultiRelaySys}
\end{figure*}


\section{System Model of Bidirectional Random Relaying}\label{Sec:SysModel}
The problem of information exchange by multihop routing can be fundamentally characterized by a bidirectional relaying system, as illustrated in Fig. \ref{Fig:TwoWayMultiRelaySys}. The two source nodes $X_A$ and $X_B$ would like to exchange their packets $W_A$ and $W_B$ over multiple relay nodes by TDMH and NC. TDMH needs four time slots to route the two packets and NC needs only three time slots due to broadcasting a XOR-ed packet $W_A\oplus W_B$ to the two source nodes. Here the nodes in the ad hoc network are assumed to form a stationary Poisson point process (PPP) of intensity $\lambda$. The network is also assumed to operate a slotted ALOHA protocol with transmission probability $p$, where $p\in(0,\frac{1}{2})$ so that the transmitters are a stationary thinning PPP of intensity $\lambda_t=\lambda\,p$, denoted by $\Phi_t=\{X_i, i\in\mathbb{N}\}$. The idle nodes (\ie nodes are not transmitting or receiving) are a stationary thinning PPP of intensity $\lambda_r=\lambda\,(1-2p)$, denoted by $\Phi_r=\{Y_j,j\in\mathbb{N}\}$\footnote{In this paper, node $X_i$ or $Y_j$ represents the node itself as well as its location in the network.}. Those idle nodes can perform like relays which are able to assist transmissions of other nodes. Specifically, we consider there exists an ``available'' \footnote{where ``available'' means any relay node in $\mathcal{D}_{AB}$ can successfully decode the information from both source nodes.} random relay regime  $\mathcal{D}_{AB}$ between nodes $X_A$ and $X_B$. Let the Lebesgue measure of $\mathcal{D}_{AB}$ denote by $\nu_r$ and thus the maximum average number of available relay nodes in $\mathcal{D}_{AB}$ is $\nu_r\,\lambda_r$. Furthermore, we assume all nodes in $\mathcal{D}_{AB}$ are able to collaborate under reasonable communication overhead so that every relay node can share its received information with others. In this context, $\mathcal{D}_{AB}$ \emph{virtually} becomes a big relay node equipped with multiple antennas so that the channels from node A to $\mathcal{D}_{AB}$ become a single-input-multiple-output (SIMO) channel (or a MISO channel from $\mathcal{D}_{AB}$ to node A).

In this work we also assume there are no direct channels between the two source nodes, otherwise, mutihop is not needed. All nodes in the network are assumed to be \emph{half-duplex} (nodes cannot transmit and receive at the same time). The fading channel gains between any two nodes $X$ and $Y$, denoted by $\{h_{XY}\}$, are independent and identically distributed (i.i.d.), reciprocal and a zero mean, circularly symmetric
complex Gaussian random variables with unit variance, $\{C_{XY}\}$ denote their corresponding channel capacities, and all transmitters have the same transmit power $\rho_0$. In order to facilitate the following descriptions and analysis, here diversity gain $d$ and multiplexing gain $m$ in \cite{LZDNCT03} need to be redefined in our notation as follows:
$$d\defn -\lim_{\bar{\gamma}\rightarrow\infty} \frac{\log \epsilon(\bar{\gamma})}{\log \bar{\gamma}}\quad \text{and}\quad m \defn \lim_{\bar{\gamma}\rightarrow\infty}\frac{R(\bar{\gamma})}{\log\bar{\gamma}},$$
where $\epsilon$ is the outage probability of \emph{information exchange in bidirectional relaying}, $R$ is the \emph{equivalent} \emph{end-to-end sum rate of two source nodes}, and $\bar{\gamma}$ is the average signal-to-interference-plus-noise (SINR) ratio without fading, which can be written as
\begin{equation}\label{Eqn:SINRwoFading}
\bar{\gamma} = \mathbb{E}[\gamma] = \mathbb{E}\left[\frac{\rho_0}{I_{\Phi_t}+N_0}\right],
\end{equation}
where $I_{\Phi_t}=\sum_{X_i\in(\Phi_t\setminus X_0)}\rho_0\,|h_{X_i}|^2\,\|X_i\|^{-\alpha}$ is the aggregate interference of a receive node that is a Poisson shot noise process\footnote{The Poisson shot noise process for receiver node $Y_j$ should be expressed as $I_{\Phi_t}=\sum_{X_i\in(\Phi_t\setminus X_0)}\rho_0\,|h_{X_iY_j}|^2\,\|X_i-Y_j\|^{-\alpha}$. Since $\Phi_t$ is stationary, according
to Slivnyak's theorem \cite{DSWKJM96} the statistics of signal reception seen by receiver $Y_j$ is the same as that seen
by any other receivers in the network. So the Poisson shot noise here is evaluated at the reference receiver located at the origin.}, $\|X_i\|$ denotes the distance between transmitter $X_i$ and the origin, and $\alpha>2$ is the path loss exponent, and $N_0$ is the noise power. Note that $\epsilon$ and $R$ are not defined based on traditional \emph{point-to-point} transmission. In this work, they are defined by an \emph{end-to-end} fashion because TDMH and NC are \emph{decode-and-forward} \emph{multihop-based} protocols. In addition, in Fig. \ref{Fig:TwoWayMultiRelaySys}(b) we call the end-to-end rate from the left node to the right node the \emph{forward} rate while the \emph{backward} rate is naturally the end-to-end rate in the opposite direction. The traffic pattern parameter $\mu$ is the ratio of the backward to the forward rate, \ie $\mu=R_{BA}/R_{AB}$.

Since the system we study here is aimed at information exchange over bidirectional relaying, it is important to ensure that the two source nodes in Fig. \ref{Fig:TwoWayMultiRelaySys} can successfully decode their packets at the same time. With this concept in mind, the reasonable way to declare an outage event happening in a bidirectional relaying system is whenever either one source node or both source nodes cannot successfully decode the desired packet. Therefore, the outage probability of transmission protocol $\texttt{S}$ for the system in Fig. \ref{Fig:TwoWayMultiRelaySys} is defined as
\begin{eqnarray}\label{Eqn:ErrProb2WrSys}
\epsilon_{\texttt{S}} \defn \mathbb{P}\left[\mathcal{E}_{\texttt{S},f}\cup \mathcal{E}_{\texttt{S},b}\right],
\end{eqnarray}
where $\mathcal{E}_{\texttt{S},f}\defn \{\tau_f\, I_{\texttt{S},f}<R_{AB}\}$ and $\mathcal{E}_{\texttt{S},b}\defn \{\tau_b\, I_{\texttt{S},b}<R_{BA}\}$ are the outage events of forward and backward transmission, and $\{\tau_f,\tau_b: \tau_f, \tau_b\in [0,1], \tau_f+\tau_b=1\}$ are time-allocation parameters for forward and backward transmission, respectively, and $\{I_{\texttt{S},f},I_{\texttt{S},b}\}$ are respectively forward and backward mutual information and studied in the following section.


\section{Mutual Information of TDMH and NC}\label{Sec:MutuInfo}
In this section we investigate the mutual information for TDMH and NC under different relay collaboration scenarios. We first start with TDMH.

\subsection{Mutual Information of TDMH}
Considering relay collaboration and a Gaussian input distribution, then the forward and backward mutual information for TDMH in a bidirectional random relaying set are shown to be
\begin{eqnarray}
I_{\texttt{TDMH} ,f} = I_{\texttt{TDMH} ,b}= \frac{1}{2}\min\left\{I_{1}, I_{2}\right\},\label{Eqn:CoopMutuInfoTDMHf}
\end{eqnarray}
where
$I_{1}\defn \log\left(1+\gamma \sum_{Y_D\in \mathcal{D}_{AB}} |h_{AD}|^2\,\|X_A-Y_D\|^{-\alpha}\right)$, $I_{2}\defn \log\left(1+\gamma \sum_{Y_D\in \mathcal{D}_{AB}} |h_{DB}|^2\,\|X_B-Y_D\|^{-\alpha}\right)$ because the forward or backward transmission first virtually passes through a SIMO channel and then through a MISO channel. Note that coefficient $\frac{1}{2}$ means the forward or backward data stream needs 2 time slots. Since all nodes have the same power and all channels are reciprocal, the forward and backward mutual information are equal.

In the case of relay noncooperation, an optimal relay node should be selected to assist bidirectional transmission by the following criterion:
\begin{equation}\label{Eqn:OptiRelTDMH}
Y_{D^*_{\texttt{TDMH}}} = \arg\max_{Y_D\in\mathcal{D}_{AB}}\left(\frac{\|X_A-Y_D\|^{\alpha}}{|h_{AD}|^2}+\frac{\|X_B-Y_D\|^{\alpha}}{|h_{BD}|^2}\right)^{-1},
\end{equation}
The selection criterion in \eqref{Eqn:OptiRelTDMH} is based on the idea of finding a relay node with the maximum end-to-end sum rates. Once $Y_{D^*_{\texttt{TDMH}}}$ is determined, its corresponding forward and backward mutual information are the same as in \eqref{Eqn:CoopMutuInfoTDMHf} with
\begin{eqnarray*}
I_1&=&\log\left(1+\gamma\,|h_{A D^*_{\texttt{TDMH}}}|^2\,\|X_A-Y_{D^*_{\texttt{TDMH}}}\|^{-\alpha}\right),\\
I_2&=&\log\left(1+\gamma\, |h_{B D^*_{\texttt{TDMH}}}|^2\,\|X_B-Y_{D^*_{\texttt{TDMH}}}\|^{-\alpha}\right).
\end{eqnarray*}
Finding an optimal relay can also provide the same diversity order due to exploited selection diversity. This result will be proved in the sequel.

\subsection{Mutual Information of NC}
For NC, its forward and backward mutual information can be shown as
\sublabon{equation}
\begin{eqnarray}
I_{\texttt{NC},f} &=& \frac{2}{3}\min\left\{I_1,\min\left\{\tilde{I}_1,\tilde{I}_2\right\}\right\},\label{Eqn:NonCoopMutuInfoDNCf}\\
I_{\texttt{NC},b} &=& \frac{2}{3}\min\left\{I_2,\min\left\{\tilde{I}_1,\tilde{I}_2\right\}\right\},\label{Eqn:NonCoopMutuInfoDNCb}
\end{eqnarray}
where $\tilde{I}_1=\log\left(1+\gamma\,|\sum_{Y_D\in\mathcal{D}_{AB}}h_{AD}\|X_A-Y_D\|^{-\alpha/2}|^2\right)$, $\tilde{I}_2=\log\left(1+\gamma\,|\sum_{Y_D\in\mathcal{D}_{AB}}h_{BD}\|X_B-Y_D\|^{-\alpha/2}|^2\right)$, and the coefficient $\frac{2}{3}$ is due to two data streams sharing three time slots. \sublaboff{equation} $\tilde{I}_1$ and $\tilde{I}_2$ stand for the mutual information for the forward broadcast channel and backward broadcast channel, respectively. They are calculated by the sum of the channel gains between $\mathcal{D}_{AB}$ and their respective destination source nodes since the relays are unable to provide the transmit diversity for both source nodes simultaneously. Accordingly, it results in a problem that the transmit diversity for both source nodes is unable to be exploited in the broadcast stage. This problem can be alleviated by using an optimal relay node to broadcast, which can be selected according to the following criterion:
\begin{eqnarray}
Y_{D^*_{\texttt{NC}}} = \arg\max_{Y_D\in\mathcal{D}_{AB}} \min\{|h_{DA}|^2\,\|X_A-Y_D\|^{-\alpha},\nonumber\\
|h_{DB}|^2\,\|X_B-Y_D\|^{-\alpha}\}.\label{Eqn:OptiBCRelayDNCwRxMRC}
\end{eqnarray}
The above criterion is to select a relay node in $\mathcal{D}_{AB}$ whose achievable broadcast channel capacity is maximal\cite{CHLFXJGA09}.

By using $Y_{D^*_{\texttt{NC}}}$ found in \eqref{Eqn:OptiBCRelayDNCwRxMRC} to broadcast, the forward and backward mutual information in \eqref{Eqn:NonCoopMutuInfoDNCf} and \eqref{Eqn:NonCoopMutuInfoDNCb} can be reduced to
$I_{\texttt{NC},f}=I_{\texttt{NC},b}=\frac{2}{3}\min\left\{\tilde{I}_1,\tilde{I}_2\right\}$ since $\tilde{I}_1=\log(1+\gamma |h_{AD_{\texttt{NC}}^*}|^2\,\|X_A-Y_{D_{\texttt{NC}}^*}\|^{-\alpha})$ and $\tilde{I}_2=\log(1+\gamma |h_{BD_{\texttt{NC}}^*}|^2\,\|X_B-Y_{D_{\texttt{NC}}^*}\|^{-\alpha})$ so that we know $I_1 \geq \tilde{I}_1$ and $I_2 \geq \tilde{I}_2$ almost surely. On the other hand, in the case of relay without collaboration what criterion we should follow to select an optimal relay node? The basic idea is also to search a relay node that can provide the maximum end-to-end sum rate. For NC, the maximum end-to-end sum rate happens whenever the bidirectional traffic is symmetric, \ie $R_{AB}=R_{BA}$\cite{CHLFXJGA09}\cite{PPHY0607}. In previous work \cite{CHLFXJGA09}, the maximum sum rate of NC over relay node $Y_D$ in terms of channel capacities is $2(1/C_{AD}+2/C_{DB})^{-1}$. So the optimal relay node $Y_{D^*_{\texttt{NC}}}$ can be equivalently selected by
\begin{equation}
Y_{D^*_{\texttt{NC}}} = \arg\max_{Y_D\in\mathcal{D}_{AB}} \left(2\frac{\|X_B-Y_D\|^{\alpha}}{|h_{BD}|^2}+\frac{\|X_A-Y_D\|^{\alpha}}{|h_{AD}|^2}\right)^{-1}.\label{Eqn:OptiRelDNC}
\end{equation}
Therefore, according to \eqref{Eqn:OptiRelDNC} the forward and backward mutual information for NC over $Y_{D^*_{\texttt{NC}}}$ can be found as
\begin{equation}\label{Eqn:NonCoopMutuInfoDNC}
I_{\texttt{NC},f}=I_{\texttt{NC},b}=\frac{2}{3}\min\{\tilde{I}_1,\tilde{I}_2\},
\end{equation}
where $\tilde{I}_{1}=\log(1+\gamma |h_{AD^*_{\texttt{NC}}}|^2\,\|X_A-Y_{D^*_{\texttt{NC}}}\|^{-\alpha})$, $\tilde{I}_{2}=\log(1+\gamma |h_{BD^*_{\texttt{NC}}}|^2\,\|X_B-Y_{D^*_{\texttt{NC}}}\|^{-\alpha}))$ and $Y_{D^*_{\texttt{NC}}}$ is determined by \eqref{Eqn:OptiRelDNC}.


\section{Main Results of DMT Analysis}\label{Sec:MainResults}
The cooperative diversity of time-division one-way relaying has been investigated in \cite{ABAKDPRAL06}\cite{JNLGWW03}. Here we investigate the DMT in bidirectional relaying for TDMH and NC. Before proceeding to the DMT analysis, we first recall the definition of an outage event happening in a bidirectional relaying system. According to \eqref{Eqn:ErrProb2WrSys} and using Boole's inequality, a bidirectional relaying system has the following inequality of outage probability:
\begin{eqnarray}
\epsilon_{\texttt{S}} \leq \epsilon_{\texttt{S},f}+\epsilon_{\texttt{S},b}\,\, ,\label{Eqn:DefnErrProb}
\end{eqnarray}
where $\texttt{S}$ means $\texttt{TDMH}$ or $\texttt{NC}$, $\epsilon_{\texttt{S},f}\defn \mathbb{P}[\mathcal{E}_{\texttt{S},f}]$ and $\epsilon_{\texttt{S},b}\defn \mathbb{P}[\mathcal{E}_{\texttt{S},f}]$. According to \eqref{Eqn:DefnErrProb}, the DMTs of TDMH and NC can be derived in the following subsections. Note that \emph{in the following analysis, we use notation $\gamma\star x$ instead of $\gamma^x$ in order to clearly present the complicated expression of exponent $x$}.

\subsection{Diversity-Multiplexing Tradeoff of TDMH}
The DMT of TDMH with or without relay collaboration is presented in the following proposition.
\begin{proposition}\label{Pro:DiveMutiTrofTDMH}
Consider $\Phi_r\cap\mathcal{D}_{AB}\neq\emptyset$ and every relay node in $\mathcal{D}_{AB}$ collaborates. TDMH achieves the following diversity-multiplexing tradeoff
\begin{equation}\label{Eqn:DiveMultTrofTDMH}
d = (\lambda_r\,\nu_r)\left(1-\frac{2m}{\min\{(1+\mu)\tau_f,(1+1/\mu)\tau_b\}}\right),
\end{equation}
where $m\in\left(0,\min\{(1+\mu)\tau_f,(1+1/\mu)\tau_b\}/2\right)$. If there is no collaboration in $\mathcal{D}_{AB}$, then TDMH over $Y_{D_{\texttt{TDMH} }^*}$ is able to achieve the DMT in \eqref{Eqn:DiveMultTrofTDMH} as well, where $Y_{D^*_{\texttt{TDMH}}}$ denotes the optimal relay node found by \eqref{Eqn:OptiRelTDMH}.
\end{proposition}
\begin{proof}
Let $\mathcal{E}_A$ ($\mathcal{E}_B$) denote the event that the relay node nodes in $\mathcal{D}_{AB}$ cannot correctly decode $W_A$ ($W_B$) and $\mathcal{E}^{c}_A$ ($\mathcal{E}^c_B$) denote the complement of $\mathcal{E}_A$ $(\mathcal{E}_B)$. Thus we have
\begin{eqnarray*}
\epsilon_{\texttt{TDMH},f} &=& \mathbb{P}\left[\mathcal{E}_{\texttt{TDMH} ,f}|\mathcal{E}_A\right]\mathbb{P}[\mathcal{E}_A]+
\mathbb{P}\left[\mathcal{E}_{\texttt{TDMH} ,f}|\mathcal{E}^c_A\right]\mathbb{P}[\mathcal{E}^c_A]\\
&=& \mathbb{P}[\mathcal{E}_A]+\mathbb{P}\left[\tau_f I_{2}<2\,R_{AB}\right]\mathbb{P}[\mathcal{E}^c_A],
\end{eqnarray*}
where $\mathbb{P}[\mathcal{E}_A]=\mathbb{P}\left[\tau_f I_{1}<2\,R_{AB}\right]$. Let $R_{AB}+R_{BA}=m\log(\gamma)$ so that $R_{AB}=\frac{m}{1+\mu}\log(\gamma)$. By using $I_1$ and $I_2$ in \eqref{Eqn:CoopMutuInfoTDMHf}, we thus have
\begin{eqnarray}
\epsilon_{\texttt{TDMH},f} &\leq& 2\,\mathbb{E}\left[\mathbb{P}\left[\min\left\{e^{I_1}, e^{I_2} \right\}-1< \gamma\star (d_f+1)|\gamma\right]\right]\nonumber\\
&\overset{(a)}{\dotleq}& \bar{\gamma}\star \left(\lambda_r\,\nu_r\,d_f\right),\label{Eqn:CoopOutageProbTDMHf}
\end{eqnarray}
for large $\bar{\gamma}$ and $m\in\left(0,\frac{1}{2}(1+\mu)\tau_f\right)$, where $d_f \defn 2\,m/\tau_f(1+\mu)-1$ and $(a)$ follows from Lemma \ref{Lem:ProbExpoRVsApprox} in Appendix. Similarly, we can show
\begin{equation}\label{Eqn:CoopOutageProbTDMHb}
\epsilon_{\texttt{TDMH},b}\dotleq \bar{\gamma}\star \left(\lambda_r\,\nu_r\,d_b \right),
\end{equation}
for large $\bar{\gamma}$ and $m\in\left(0,\frac{1}{2}(1+1/\mu)\tau_b\right)$, where $d_b\defn 2m/(1+1/\mu)\tau_b-1$. According to \eqref{Eqn:DefnErrProb}, it thus follows that
\begin{equation*}
\epsilon_{\texttt{TDMH}} \dotleq \bar{\gamma}\star \left[\lambda_r\,\nu_r\,\left(\frac{2m}{\min\{(1+\mu)\tau_f,(1+1/\mu)\tau_b\}}-1\right)\right],
\end{equation*}
for large $\bar{\gamma}$ and $m\in(0,\frac{1}{2}\min\{(1+\mu)\tau_f,(1+1/\mu)\tau_b\})$.

Now consider there is no collaboration in $\mathcal{D}_{AB}$. The optimal relay node $Y_{D_{\texttt{TDMH} }^*}$ is selected according to \eqref{Eqn:OptiRelTDMH}. So we can obtain
\begin{eqnarray*}
\epsilon_{\texttt{TDMH},f} \leq 2\,\mathbb{P}[\min\{|h_{A D_{\texttt{TDMH}}^*}|^2\|X_A-Y_{D_{\texttt{TDMH}}^*}\|^{-\alpha},\\
|h_{BD^*_{\texttt{TDMH} }}|^2\|X_B-Y_{D_{\texttt{TDMH}}^*}\|^{-\alpha}\}< \gamma\star d_f].
\end{eqnarray*}
Since $Y_{D^*_{\texttt{TDMH}}}$ is optimal in $\mathcal{D}_{AB}$ and all channels are independent, we further have
\begin{eqnarray*}
\epsilon_{\texttt{TDMH},f} \leq \hspace{3in}\text{ }\\  2\prod_{Y_D\in\mathcal{D}_{AB}}\mathbb{P}\left[\left(\frac{\|X_A-Y_D\|^{\alpha}}{|h_{AD}|^2}+\frac{\|X_B-Y_D\|^{\alpha}}{|h_{BD}|^2}\right)^{-1} <\gamma\star d_f\right]\\
\overset{(b)}{\dotleq} \bar{\gamma}\star \left(\lambda_r\,\nu_r\,d_f \right),\hspace{0.1in}\text{ }
\end{eqnarray*}
where $(b)$ follows from Lemma \ref{Lem:ProbOptiRandVectApprox} in Appendix and Campbell's theorem\cite{DSWKJM96}. Likewise, we can get a similar result for $\epsilon_{\texttt{TDMH},b}$ as shown in \eqref{Eqn:CoopOutageProbTDMHb}. Thus optimal relay selection achieves the same DMT with relay collaboration in \eqref{Eqn:DiveMultTrofTDMH}.
\end{proof}

\subsection{Diversity-Multiplexing Tradeoff of NC}
Using NC in bidirectional multi-relaying has three transmission scenarios. If all relay nodes collaborate, in the first two time slots NC can have receive diversity at $\mathcal{D}_{AB}$ and no transmit diversity in the third time slot if all relay nodes join to broadcast. A better strategy in this case is to select an optimal relay to broadcast. For relay without collaboration, an optimal relay should be found to route packets. The DMTs of NC with these scenarios have been presented in the following proposition.
\begin{proposition}\label{Pro:DiveMutiTrafNC}
Suppose $\Phi_r\cap\mathcal{D}_{AB}\neq\emptyset$ and all relay nodes in $\mathcal{D}_{AB}$ collaborate to receive and then broadcast at the same time. The following DMT is achieved by NC:
\begin{equation}\label{Eqn:DiveMultTrofNC01}
d=1-\frac{3m}{2\min\{(1+\mu)\tau_f,(1+1/\mu)\tau_b\}},
\end{equation}
where $m\in\left(0,\frac{2}{3}\min\{(1+\mu)\tau_f,(1+1/\mu)\tau_b\}\right)$. If an optimal relay node is selected by \eqref{Eqn:OptiBCRelayDNCwRxMRC} to broadcast, NC achieves the following diversity-multiplexing tradeoff:
\begin{equation}\label{Eqn:DiveMultTrofNC02}
d=(\lambda_r\,\nu_r)\left(1-\frac{3m}{2\min\{(1+\mu)\tau_f,(1+1/\mu)\tau_b\}}\right).
\end{equation}
Furthermore, if an optimal relay node is selected to receive and broadcast then the DMT in \eqref{Eqn:DiveMultTrofNC02} is achieved as well.
\end{proposition}

\begin{proof}
By the definition of outage and using the same definitions of $\mathcal{E}_A$ and $\mathcal{E}_B$ in the proof of Proposition \ref{Pro:DiveMutiTrofTDMH}. So the outage probability of forward transmission can be shown as
\begin{eqnarray*}
\epsilon_{\texttt{NC},f} &\leq& \mathbb{P}[\mathcal{E}_A]
+\mathbb{P}\left[2\tau_f\min\{\tilde{I}_{2},\tilde{I}_{1}\}<3R_{AB}\right]\\
&\leq& \mathbb{P}[\mathcal{E}_A]+\mathbb{P}\left[2\tau_f\tilde{I}_{2}<3R_{AB}\right]+
\mathbb{P}\left[2\tau_f\tilde{I}_{1}<3R_{AB}\right],
\end{eqnarray*}
where $\mathbb{P}[\mathcal{E}_A]=\mathbb{P}\left[\frac{2}{3}\tau_f I_{1}<R_{AB}\right]$. Let $R_{AB}=\frac{m}{1+\mu}\log(\gamma)$ and consider the first case that every relay node collaborates to receive and then broadcasts without collaboration simultaneously. For large $\bar{\gamma}$ and using Lemma \ref{Lem:ProbExpoRVsApprox}, it follows that
\begin{eqnarray}
\epsilon_{\texttt{NC},f} &\leq & \bar{\gamma}\star \left(\lambda_r\,\nu_r\,\tilde{d}_f\right)+\left(\Xi_{b_1}+\Xi_{b_2} \right)\left(\bar{\gamma}\star \tilde{d}_f\right)\nonumber\\
&\dotleq& \bar{\gamma}\star \tilde{d}_f,
\end{eqnarray}
for large $\bar{\gamma}$ and $m\in(0,2(1+\mu)\tau_f/3)$, where $1/\Xi_{b_1}$ and $1/\Xi_{b_2}$ are respectively the variances of $|\sum_{Y_D\in\mathcal{D}_{AB}}h_{DA}\,\|X_A-Y_D\|^{-\alpha/2}|^2$ and $|\sum_{Y_D\in\mathcal{D}_{AB}}h_{DB}\,\|X_B-Y_D\|^{-\alpha/2}|^2$, and $\tilde{d}_f\defn 3m/2(1+\mu)\tau_f-1$. Similarly, we can show $\epsilon_{\texttt{NC},b}\dotleq \bar{\gamma}\star \tilde{d}_b$, for large $\bar{\gamma}$ and $m\in(0,2(1+1/\mu)\tau_b/3)$, where $\tilde{d}_b \defn 3m/2(1+1/\mu)\tau_b-1$. Then \eqref{Eqn:DiveMultTrofNC01} can be obtained since $\epsilon_{\texttt{NC}} \leq \epsilon_{\texttt{NC},f}+\epsilon_{\texttt{NC},b}$.

Consider NC with optimal relay $Y_{D_{\texttt{NC}}^*}$ selected by \eqref{Eqn:OptiBCRelayDNCwRxMRC} to broadcast. Then we have
\begin{eqnarray*}
&{}&\mathbb{P}\left[2\tau_f\tilde{I}_{2}<3R_{AB}\right]\nonumber\\
&&= \mathbb{E}\left[\mathbb{P}\left[|h_{D_{\texttt{NC}}^* B}|^2\|X_B-Y_{D_{\texttt{NC}}^*}\|^{-\alpha}<\gamma\star\tilde{d}_f\bigg|\gamma\right]\right]\nonumber\\
&&\overset{(a)}{\dotleq} \bar{\gamma}\star (\lambda_r\,\nu_r\,\tilde{d}_f),
\end{eqnarray*}
where $(a)$ follows that $Y_{D^*_{\texttt{NC}}}$ is optimal and $\{h_{DB}\}$ are independent, and from Lemma \ref{Lem:ProbExpoRVsApprox} in Appendix. Similarly, we have
$$\mathbb{P}\left[2\tau_b\tilde{I}_{1}<3R_{BA}\right]\dotleq \bar{\gamma}\star\left(\lambda_r\,\nu_r\,{\tilde{d}_b}\right).$$
Therefore, we can conclude
\begin{eqnarray*}
\epsilon_{\texttt{NC}} \dotleq \bar{\gamma}\star \left[\lambda_r\,\nu_r\,\left(\frac{3m}{2\min\{(1+\mu)\tau_f,(1+1/\mu)\tau_b\}}-1\right)\right].
\end{eqnarray*}

Next, we are going to look at the DMT of NC using an optimal relay node $Y_{D^*_{\texttt{NC}}}$ to receive and broadcast. $Y_{D^*_{\texttt{NC}}}$ is determined by \eqref{Eqn:OptiRelDNC}. Likewise, the first step is to calculate $\epsilon_{\texttt{NC},f}$ by \eqref{Eqn:NonCoopMutuInfoDNC} with $R_{AB}=\frac{m}{1+\mu}\log(\gamma)$, and thus we know the forward outage probability \eqref{Eqn:OutageProbFwadNC} shown on the top of the next page.
\begin{figure*}
\begin{equation}\label{Eqn:OutageProbFwadNC}
\mathbb{P}\left[\tilde{I}_1<\frac{3R_{AB}}{2\,\tau_f}\right] \leq \mathbb{E}\left[\mathbb{P}\left[\frac{|h_{AD_{\texttt{NC}}^*}|^2|h_{D^*_{\texttt{NC}}B}|^2}{2|h_{AD^*_{\texttt{NC}}}|^2
\|X_B-Y_{D^*_{\texttt{NC}}}\|^{\alpha}+|h_{D^*_{\texttt{NC}}B}|^2\|X_A-Y_{D^*_{\texttt{NC}}}\|^{\alpha}}< \gamma\star \left(\tilde{d}_f-1\right)\bigg|\gamma\right]\right].
\end{equation}
\end{figure*}
So we have $\epsilon_{\texttt{NC},f}\dotleq\bar{\gamma}\star (\lambda_r\,\nu_r\,\tilde{d}_f)$ for large $\bar{\gamma}$ and $m\in(0,2(1+\mu)\tau_f/3)$
because $Y_{D^*_{\texttt{NC}}}$ is optimal, and channel gains are independent so that Lemma \ref{Lem:ProbOptiRandVectApprox} in Appendix can be applied. Similarly, the exponential inequality for $\epsilon_{\texttt{NC},b}$ is $\epsilon_{\texttt{NC},b} \dotleq \bar{\gamma}\star (\lambda_r\,\nu_r\,\tilde{d}_b)$, for large $\bar{\gamma}$ and $m\in(0,2(1+1/\mu)\tau_b/3)$, where $\tilde{d}_b\defn 3m/2(1+1/\mu)\tau_b-1$. Thus NC over an optimal relay node achieves the DMT same as indicated in \eqref{Eqn:DiveMultTrofNC02}.
\end{proof}
The results in Propositions \ref{Pro:DiveMutiTrofTDMH} and \ref{Pro:DiveMutiTrafNC} have been presented in Fig. \ref{Fig:DMTradeoff} for $\mu=1$. For the case of $\tau_f=\tau_b=0.5$ in the figure, NC always has a better DMT than TDMH when relay nodes collaborate to receive and an optimal relay is selected to broadcast. This is because relay selection diversity is exploited to broadcast. If all relay nodes broadcast, NC will loose diversity since it is hard to achieve bidirectional transmit diversity at the same time for the relays in $\mathcal{D}_{AB}$. NC does not necessarily have a better DMT than TDMH if $\{\tau_f,\tau_b\}$ are  not optimally assigned. For example, if the forward and reverse times between node $X_A$ and $\mathcal{D}_{AB}$ are 0.01, the forward and reverse times between $\mathcal{D}_{AB}$ and node $X_B$ are 0.49 and $\mu=1$ then TDMH has $\tau_f=\tau_b=(0.01+0.49)/(0.5+0.5)=0.5$ and its DMT is $\lambda_r\,\nu_r\,(1-2m)$ while NC has $\tau_f=\frac{0.01+0.49}{(0.01+0.49)+2\cdot 0.49}\approx 0.34$ and $\tau_b=1-0.34=0.66$ and its DMT in \eqref{Eqn:DiveMultTrofNC02} becomes $d=\lambda_r\,\nu_r\,(1-2.2m)$. So NC has a worse DMT than TDMH in this case. Furthermore, the ideal DMT can be asymptotically approached if network coding can support information exchange for $N$ source nodes within $N+1$ time slots even when $N$ is very large.

\begin{figure}[h]
\centering
\includegraphics[width=3.5in, height=1.8in]{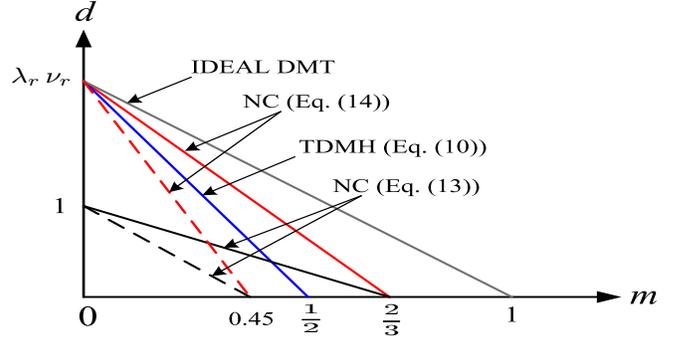}\\
\caption{Diversity-multiplexing tradeoffs for different transmission protocols ($\lambda_r\,\nu_r>1$, $\mu=1$). The results of solid lines are the case of optimal time allocation for NC, \ie $\tau_f=\tau_b=0.5$. The results of dashed lines are the case of suboptimal time allocation of NC, \ie $\tau_f=0.34$ and $\tau_b=0.66$.}\label{Fig:DMTradeoff}
\end{figure}

\section{Simulation Results}
From the above results in Propositions \ref{Pro:DiveMutiTrofTDMH} and \ref{Pro:DiveMutiTrafNC}, the DMT achieved by NC would be worse than that achieved by TDMH if the time allocation between forward and backward traffic is suboptimal. Here we simulate the DMT case that the two-way traffic of the two protocols is respectively through their optimal relay. We assume that all nodes have the same transmit power 18 dBm, and the channel between any two nodes has path loss exponent 3.5 and is reciprocal with flat Rayleigh fading. The distance between source nodes A and B is 60m, and the random relaying set is a circular area which has a diameter of 10m and is centered at the middle point between nodes A and B.

Suppose the node intensity $\lambda=0.1$, traffic pattern parameter $\mu=1$, multiplexing gain $m=1/4$. Consider the relays in the random set are not cooperative and an optimal relay is selected for routing/broadcasting the packets. The simulation results of outage probability versus average SINR for optimal and suboptimal time allocation are shown in Figs. \ref{Fig:OutProbST} and \ref{Fig:OutProbAT}, respectively. We can see the diversity gain of NC is almost the same as that of TDMH in Fig. \ref{Fig:OutProbAT}, while in Fig. \ref{Fig:OutProbST} the diversity gain of NC is obviously superior to that of TDMH. Therefore, from the DMT point of view one can also show that NC may not be always superior to TDMH when time allocations for bidirectional traffic are suboptimal and/or bidirectional traffic is asymmetric.

\begin{figure}[h]
\centering
\includegraphics[width=3.6in, height=2.6in]{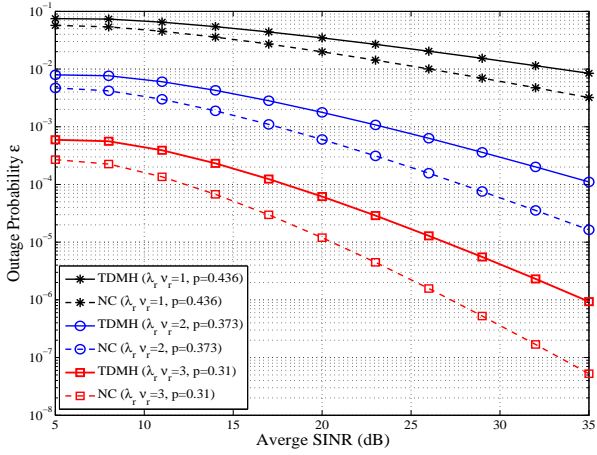}\\
\caption{Outage probabilities of the TDMH and NC protocols without relay collaboration. An optimal relay node is
selected to receive and transmit/broadcast for the two protocols and time allocation for bidirectional traffic is optimal, \ie $\tau_f=\tau_b=0.5$.}\label{Fig:OutProbST}
\end{figure}

\begin{figure}[h]
\centering
\includegraphics[width=3.6in, height=2.6in]{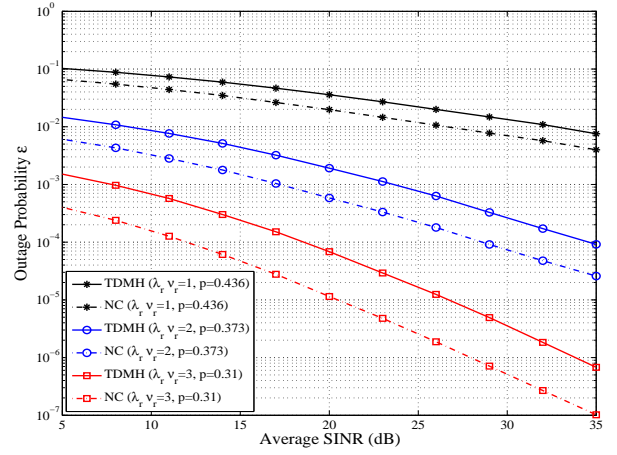}\\
\caption{Outage probabilities of the TDMH and NC protocols without relay collaboration. An optimal relay node is
selected to receive and transmit/broadcast for the two protocols and time allocation of bidirectional traffic is suboptimal for NC ($\tau_f=0.34$ and $\tau_b=0.66$) and optimal for TDMH ($\tau_f=\tau_b=0.5$).}\label{Fig:OutProbAT}
\end{figure}


\section{Concluding Remarks}\label{Sec:Conclusion}
The DMTs of TDMH and NC in the different scenarios of relay collaboration have been investigated in this paper. The information exchange between the two source nodes is over a random relay set in which the distribution of the relays is a stationary PPP. The DMT analysis here is based on end-to-end bidirectional outage so that the DMTs are affected by traffic pattern, time allocation between bidirectional traffic as well as the average number of relay nodes in the random relay set. Our main result proves that NC does not always have a better DMT than TDMH in the relay collaboration case because bidirectional transmit diversity is difficult to be achieved for both source nodes at the same time. In addition, the DMT of NC could be worse than that of TDMH as well due to suboptimal time allocation between the bidirectional traffic. From the DMT results, we can obtain some insight of how to do time sharing between the bidirectional traffic to achieve a better DMT for a given traffic pattern.


\appendix [Lemmas for DMT Analysis]
\begin{definition}\label{Defn:ExpOrder}
A function $g(\omega):\mathbb{R}_{++}\rightarrow\mathbb{R}_{++}$ is said to exponentially smaller than or equal to $x$, \ie $g(\omega)\dotleq\omega^x$, if $\lim_{\omega\rightarrow\infty}\log g(\omega)/\log\omega \leq x$. Similar definition can be applied to the equal sign.
\end{definition}
\begin{lemma}\label{Lem:ProbExpoRVsApprox}
Let $\mathcal{B}_z\subseteq\mathbb{R}^2$ be a Borel set and $\Phi'_z\defn\{(Z_i,g_i): i\in\mathbb{N}\}$ be a marked stationary PPP of intensity $\lambda_z$ where $\{g_i\}$ are i.i.d. exponential random variables with unit mean and variance. The distance between node $Z_i$ and the origin denotes by $\|Z_i\|$ and $\theta(\omega):\mathbb{R}_{++}\rightarrow \mathbb{R}_{++}$. If $\theta(\omega)\rightarrow 0$ as $\omega\rightarrow\infty$ and $\theta(\omega)$ is exponentially equal to $\theta_{\infty}$, then we have
\begin{equation}\label{Eqn:IeqExpRVsExpRel}
\mathbb{P}\left[\sum_{Z_k\in \Phi_z}g_k\,\|Z_k\|^{-\alpha} < \theta(\omega)\right]\dotleq \omega\star(\lambda_z\,\nu_z\,\theta_{\infty}),
\end{equation}
where $\alpha>2$, $\Phi_z=\Phi'_z\cap\mathcal{B}_z$, and $\nu_z$ is the Lebesgue measure of $\mathcal{B}_z$.
\end{lemma}
\begin{proof}
Without loss of generality, assume that the finite random sequence $\{g_k\,\|Z_k\|^{-\alpha}: Z_k\in\mathcal{B}_z, k\in\mathbb{N}_+\}$ forms an order statistics, \ie $\{g_1\,\|Z_1\|^{-\alpha}\leq g_2\,\|Z_2\|^{-\alpha}\leq g_3\,\|Z_3\|^{-\alpha}\cdots\leq g_k\,\|Z_k\|^{-\alpha}\leq\cdots\}$. Thus, the event $\sum_{Z_k\in\Phi_z} g_k \|Z_k\|^{-\alpha} \leq \theta(\omega)$ is equivalent to the intersection event of $g_1\,\|Z_1\|^{-\alpha}\leq \theta(\omega)$, $g_1\,\|Z_1\|^{-\alpha}+g_2\,\|Z_2\|^{-\alpha}\leq\theta(\omega),\cdots,\sum_{Z_k\in\Phi_z} g_k\,\|Z_k\|^{-\alpha} \leq \theta(\omega)$. Hence,
\begin{eqnarray*}
&{}&\log\left\{\mathbb{P}\left[\sum_{Z_k\in\Phi_z} g_k\,\|Z_k\|^{-\alpha} \leq \theta(\omega)\bigg|\Phi_z\right]\right\}\\ &&= \log\left\{\mathbb{P}\left[\bigcap_{Z_k\in\Phi_z}\left(\sum_{j=1}^{k}g_j\,\|Z_j\|^{-\alpha}\leq\theta(\omega)\right)\bigg|\Phi_z\right]\right\}\\
&&\leq \log\left\{\mathbb{P}\left[\bigcap_{Z_k\in\Phi_z} \left(g_k\,\|Z_k\|^{-\alpha}\leq\theta(\omega)\right)\bigg|\Phi_z\right]\right\}\\
&&\overset{(a)}{=} \sum_{Z_k\in\Phi_z}\log\left\{\mathbb{P}\left[g_k\,\|Z_k\|^{-\alpha}\leq\theta(\omega)\bigg|\Phi_z\right]\right\},
\end{eqnarray*}
where $(a)$ follows from the independence between all random variables. Since all random variables are exponential, then we further have
\begin{eqnarray}\label{Eqn:CondCDFChannelGain}
\mathbb{P}\left[g_k\,\|Z_k\|^{-\alpha}\leq\theta(\omega)\bigg|\Phi_z\right]&=& 1-\exp\left(-\|Z_k\|^{\alpha}\theta(\omega)\right)\nonumber\\
&\overset{(b)}{\leq}& \|Z_k\|^{\alpha}\theta(\omega),
\end{eqnarray}
where $(b)$ follows from the fact that $g_k$ is an exponential random variable with unit variance and $e^{-y}\geq 1-y,\,\forall y\in\mathbb{R}_+$.
Using \eqref{Eqn:CondCDFChannelGain} and letting $\mathcal{B}_z$ be outer bounded by a minimum disc of radius $s$, then we have
\begin{eqnarray*}
&{}&\mathbb{P}\left[\sum_{Z_k\in \Phi_z}g_k\,\|Z_k\|^{-\alpha} < \theta(\omega)\right]\\
&&\leq \exp\left\{\mathbb{E}\left[\sum_{Z_k\in\Phi_z} \log(\theta(\omega))-\alpha\log(\|Z_k\|)\right]\right\}\\
&&\overset{(c)}{=} \left(\theta(\omega)\,s^{-\alpha}\right)\star(\lambda_z\,\nu_z),
\end{eqnarray*}
where $(c)$ follows from Campbell's theorem\cite{DSWKJM96}. By Definition \ref{Defn:ExpOrder} and $\theta(\omega)\doteq \omega^{\theta_{\infty}}$, the result in \eqref{Eqn:IeqExpRVsExpRel} is readily obtained.
\end{proof}

\begin{lemma}\label{Lem:ProbOptiRandVectApprox}
Let $\mathcal{T}$ be a given countable finite set with cardinality $|\mathcal{T}|$ and $\mathcal{V}$ be a random vector set whose elements are $m$-tuples, independent and nonnegative, \ie $\mathcal{V}\defn\{\mathbf{V}_i, i\in\mathbb{N}_{+}: \mathbf{V}_i\in\mathbb{R}_+^m, \mathbf{V}_i\bot\mathbf{V}_j, i\neq j\}$. Suppose $\forall t\in\mathcal{T}$, $\mathbf{V}_t=(V_{t_1},V_{t_2},\ldots,V_{t_m})^{\top}\in\mathcal{V}$ is an exponential random vector with $m$ independent entries and $\omega\in\mathbb{R}_{++}$. Suppose $t^* \defn \arg\max_{t\in\mathcal{T}} f(\mathbf{V}_t)$ where $f(\mathbf{V}_t)$ is defined as
\begin{equation}
   f(\mathbf{V}_t) \defn \frac{\prod_{i=1}^m V_{t_i}}{\sum_{i=1}^m \beta_i(\omega) (V_{t_i})^m},
\end{equation}
where $\{\beta_i(\omega)\in\mathbb{R_{++}}\}$ are exponentially equal to $\{\beta_{i_{\infty}}\}$. If $\theta(\omega)$ is exponentially equal to $\theta_{\infty}$ and $\theta(\omega)\rightarrow 0$ as $\omega\rightarrow\infty$, then for sufficient large $\omega$ we have
\begin{equation}
    \mathbb{P}\left[f(\mathbf{V}_{t^*})<\theta(\omega)\right]\dotleq \omega^{|\mathcal{T}|(\theta_{\infty}+m\beta^+_{\max})},
\end{equation}
where $\beta^{+}_{\max} \defn \max_{i}\{\beta_{i_{\infty}},0\}$.
\end{lemma}
\begin{proof}
Since we know all random vectors in $\mathcal{V}$ are independent and $t^* = \arg\max_{t\in\mathcal{T}} f(\mathbf{V}_t)$, we have
\begin{eqnarray}\label{Eqn:ProbRandVect01}
\mathbb{P}\left[f(\mathbf{V}_{t^*})<\theta(\omega)\right]=\prod_{t\in\mathcal{T}}\mathbb{P}\left[f(\mathbf{V}_t)<\theta(\omega)\right].
\end{eqnarray}
In addition, for any $t\in\mathcal{T}$ it is easy to show that
$$f(\mathbf{V}_t) \geq \phi_m(\omega) \frac{V_{t_{\max}}}{V_{t_{\min}}} \left(\frac{V_{t_{\min}}}{V_{t_{\max}}}\right)^{m+1} = \phi_m(\omega) V_{t_{\max}} \Psi_{t_m},$$
where $\phi_m(\omega) \defn 1/[1+\sum_{i=1}^m\beta_i(\omega)]$, $V_{t_{\min}} \defn \min\{\mathbf{V}_t\}$, $V_{t_{\max}} \defn \max\{\mathbf{V}_t\}$, $\Psi_{t_m} \defn (V_{t_{\min}})^m/(V_{t_{\max}})^{m+1}$. Thus, $\mathbb{P}\left[f(\mathbf{V}_t) < \theta(\omega)\right] \leq \mathbb{P}\left[V_{t_{\max}} \Psi_{t_m}< \phi^{-1}_m \theta(\omega)\right]$. Also, we know
\begin{eqnarray*}
\mathbb{P}\left[V_{t_i}\Psi_{t_m} < \phi^{-1}_m \theta(\omega)\right] \overset{(a)}{\leq}
\int_{\mathbb{R}_{++}} \frac{\sigma_{t_i}\theta(\omega)}{\phi_m(\omega)\psi_{t_m}}f_{\Psi_{t_m}}(\psi_{t_m})\, d\psi_{t_m},
\end{eqnarray*}
where $f_{\Psi_{t_m}}(\psi_{t_m})$ is the probability density function of $\Psi_{t_m}$ and $(a)$ follows from exponential random variable $V_{t_i}$ with parameter $\sigma_{t_i}$ and $e^{-x} \geq 1-x,\,\,\forall x\in\mathbb{R}_+$. So for large $\omega$, we can obtain
\begin{eqnarray*}
\mathbb{P}\left[f(\mathbf{V}_t) < \theta(\omega)\right] &\leq& \prod^m_{i=1}\frac{\sigma_{t_i}\theta(\omega)}{\phi_m}\mathbb{E}\left[\frac{1}{\Psi_{t_m}}\right]\\
&\leq&\Sigma_t [\phi^{-1}_m(\omega)]^{m} \theta(\omega),
\end{eqnarray*}
where $\Sigma_t \defn \left(\mathbb{E}[1/\Psi_{t_m}]\right)^m \prod_{i=1}^{m}\sigma_{t_i}$. So \eqref{Eqn:ProbRandVect01} becomes
\begin{eqnarray*}
\mathbb{P}\left[f(\mathbf{V}_{t^*})<\theta(\omega)\right]&\leq& [\Sigma_t (\phi_m(\omega))^{-m}\theta(\omega)]^{|\mathcal{T}|}.
\end{eqnarray*}
For large $\omega$, $\mathbb{P}\left[f(\mathbf{V}_{t^*})<\theta(\omega)\right] \dotleq \omega^{|\mathcal{T}|(\theta_{\infty}+m\beta^+_{\max})}$.
\end{proof}

\bibliographystyle{ieeetran}
\bibliography{IEEEabrv,Ref_DmtMnc2Wr}

\end{document}